\documentclass[letterpaper, 10 pt, conference]{ieeeconf}

\IEEEoverridecommandlockouts                              

\usepackage{changes}
\usepackage{authblk}
\usepackage{amsfonts}

\usepackage{amsthm}
\usepackage{amsmath}
\usepackage{amssymb}
\usepackage{hyperref,cite}
\usepackage{mathtools, nccmath}
\usepackage{cleveref}
\usepackage{color}
\usepackage{algorithm}
\usepackage{subcaption}
\usepackage{algorithmic}
\let\labelindent\relax
\usepackage{enumitem}
\usepackage{float}
\usepackage{graphicx}
\usepackage{booktabs}
\usepackage{threeparttable}
\usepackage{hhline}
\newtheorem{theorem}{Theorem}

\newtheorem{definition}{Definition}
\newtheorem{remark}{Remark}
\newtheorem{problem}{Problem}

\newcommand{\rr}{{\mathbb R}}
\newcommand{\ba}[1]{\begin{array}{#1}}
\newcommand{\ea}{\end{array}}

\begin{document}
\title{\LARGE \bf
Koopman-Based Linear MPC for Safe Control \\ using Control Barrier Functions}
\author{Shuo Liu$^{1,*}$, Liang Wu$^{2,*}$, Dawei Zhang$^{1}$, J\'an Drgo\v na$^2$ and Calin A. Belta$^{3}$%
\thanks{This work was supported in part by the NSF under grant IIS-2024606 at Boston University and by a Brendan Iribe endowed professorship at the University of Maryland. 
This research was also supported by the U.S. DOE, Office of Science, ASCR program under the Scientific Discovery through Advanced Computing (SciDAC) Institute “LEADS: LEarning-Accelerated Domain Science”.  }
\thanks{$^{1}$S. Liu and D. Zhang are with the Department of Mechanical Engineering, Boston
University, Brookline, MA, USA {\tt\small \{liushuo, dwzhang\}@bu.edu}}
\thanks{$^{2}$L. Wu and J. Drgona are with the Department of Civil and Systems Engineering, Johns Hopkins University, Baltimore, MD, USA 
        {\tt\small \{wliang14, jdrgona1\}@jh.edu}}
\thanks{$^{3}$C. Belta is with the Department of Electrical and Computer Engineering and with the Department of Computer Science, University of Maryland, College Park, MD, USA 
        {\tt\small calin@umd.edu}}%
\thanks{$^{*}$These authors contributed equally to this work.}%
}

\maketitle

\begin{abstract}
This paper proposes a Koopman-based linear model predictive control (LMPC) framework for safety-critical control of nonlinear discrete-time systems. Existing MPC formulations based on discrete-time control barrier functions (DCBFs) enforce safety through barrier constraints but typically result in computationally demanding nonlinear programming. To address this challenge, we construct a DCBF-augmented dynamical system and employ Koopman operator theory to lift the nonlinear dynamics into a higher-dimensional space where both the system dynamics and the barrier function admit a linear predictor representation. This enables the transformation of the nonlinear safety-constrained MPC problem into a quadratic program (QP). To improve feasibility while preserving safety, a relaxation mechanism with slack variables is introduced for the barrier constraints. The resulting approach combines the modeling capability of Koopman operators with the computational efficiency of QP. Numerical simulations on a navigation task for a robot with nonlinear dynamics demonstrate that the proposed framework achieves safe trajectory generation and efficient real-time control.
\end{abstract}


\section{Introduction}
Ensuring safety is a fundamental requirement in many autonomous systems such as autonomous cars, robots, and aerial vehicles. These systems must operate in complex environments while satisfying safety requirements such as collision avoidance and input limits. Control Barrier Functions (CBFs) \cite{ames2014control, ames2016control} have emerged as an effective tool for enforcing safety requirements in dynamical systems. By constructing a barrier function whose superlevel set defines the safe region, CBFs guarantee forward invariance of the safe set by imposing affine constraints on the control input. Moreover, stabilization and safety can be addressed simultaneously by combining CBFs with Control Lyapunov Functions (CLFs) through quadratic programs (QPs). Extensions of this framework have been developed to handle high-relative-degree constraints and adaptive control scenarios \cite{nguyen2016exponential, xiao2021high, xiao2021adaptive, liu2023auxiliary, liu2024auxiliary}. For discrete-time systems, discrete-time CBFs (DCBFs) were introduced in \cite{agrawal2017discrete}.

The aforementioned approaches are, in general, shortsighted and aggressive due to the lack of predicting ahead. Model predictive control (MPC) with DCBFs \cite{zeng2021enhancing} addresses safety in the discrete-time domain and produces smooth control actions by incorporating future state information over a receding horizon. However, the computational cost grows significantly with the horizon length, as the resulting optimization is typically nonlinear and nonconvex. To improve efficiency and handle high relative-degree constraints, discrete-time high-order CBFs (DHOCBFs) were introduced within an iterative MPC framework \cite{liu2023iterative}. However, the framework still requires successive linearizations of the nonlinear dynamics and CBF constraints, and the convergence of the resulting iterations is not always guaranteed.

An alternative approach is provided by the Koopman operator framework \cite{rowley2009spectral,budivsic2012applied}, which represents nonlinear dynamics as linear systems in a lifted state space. Data-driven methods such as Dynamic Mode Decomposition (DMD) and extended DMD (EDMD) enable efficient finite-dimensional approximations of the Koopman operator \cite{schmid2010dynamic, williams2015data}. Recently, Koopman-based methods have also been explored in CBF-based safety control. For example, \cite{folkestad2020data} leverages Koopman operators to propagate nonlinear dynamics for efficient safety set computation, \cite{zinage2023neural} learns Koopman model and corresponding CBFs for unknown nonlinear systems using neural networks, and \cite{black2023safe} develops a Koopman-based identification scheme with fixed-time convergence for safe control of uncertain nonlinear systems. Despite these advances, existing Koopman-CBF approaches do not incorporate prediction over a receding horizon. Integrating Koopman models with MPC enables such prediction capability \cite{korda2018linear}. However, to maintain convexity in the resulting optimization problem, the CBF constraints still require additional linearization.

 In this paper, we propose a Koopman-based MPC-DCBF framework that enables computationally efficient safety-critical control for nonlinear systems. To the best of our knowledge, this work is the first to integrate Koopman-based modeling, MPC prediction, and CBF safety constraints within a unified MPC framework. In particular, the contributions are as follows:

\begin{itemize}
    \item We propose a Koopman-based MPC-DCBF framework for safety-critical control of nonlinear systems by lifting a DCBF-augmented dynamical system into a higher-dimensional linear predictor representation, leading to a linear MPC formulation that can be efficiently solved as a convex quadratic program.

    \item We introduce a relaxation mechanism with slack variables for the DCBF constraints, preserving safety guarantees while improving feasibility and maintaining the quadratic programming structure of the MPC problem.

    \item Numerical simulations on a nonlinear mobile robot navigation task demonstrate that the proposed approach achieves safe trajectory generation while significantly reducing computational cost compared with existing NMPC-DCBF and iterative MPC-DCBF methods.
\end{itemize}

The remainder of the article is organized as follows. In
Sec. \ref{sec: preliminaries}, we provide definitions and preliminaries. We formulate the problem and outline our approach in Sec. \ref{sec: problem and approach}. The proposed Koopman-based MPC-DCBF framework is introduced in Sec. \ref{sec: LMPC-DCBF} followed by
simulations in Sec. \ref{sec: case studies}. We conclude the paper and discuss directions for future work in Sec. \ref{sec:conclusion}.

\section{Preliminaries}
\label{sec: preliminaries}
\subsection{Koopman Operator}\label{subsec: EDMD}



Koopman theory \cite{koopman1931hamiltonian,koopman1932dynamical} provides an operator-theoretic framework that lifts an autonomous nonlinear system to a linear but infinite-dimensional representation. Extensions to controlled systems \cite{williams2016extending,proctor2018generalizing,korda2018linear} consider
\begin{equation}\label{eq: sys}
\boldsymbol{x}_{t+1}=f(\boldsymbol{x}_t,\boldsymbol{u}_t),
\end{equation}
where $\boldsymbol{x}_t \in \mathcal{X} \subset \mathbb{R}^n$ is the system state at time step $t \in \mathbb{N}$, $\boldsymbol{u}_t \in \mathcal{U} \subset \mathbb{R}^q$ is the control input, and $f: \mathbb{R}^n \times \mathbb{R}^q \to \mathbb{R}^n$ denotes the system dynamics.
A common approach to generalizing the Koopman operator to controlled systems 
\cite{korda2018linear} introduces an extended state 
$\boldsymbol{\xi}_t = \begin{bmatrix} \boldsymbol{x}_t \\ \mathbf{u}_t \end{bmatrix}$, 
where $\boldsymbol{x}_t \in \mathcal{X}$ denotes the state at time $t$ and 
$\mathbf{u}_t = \{\boldsymbol{u}_{t+k}\}_{k=0}^{\infty} \in \ell(\mathcal{U})$ 
is the sequence of future inputs starting from time $t$, with $\boldsymbol{u}_{t+k} \in \mathcal{U}$, and $\ell(\mathcal{U})$ denotes the space of all such admissible input sequences. 
The dynamics of the extended state $\boldsymbol{\xi}_t$ is described by
$\boldsymbol{\xi}_{t+1}=f_{\boldsymbol{\xi}}(\boldsymbol{\xi}_t) = 
\begin{bmatrix} f(\boldsymbol{x}_t, \boldsymbol{u}_t) \\ \boldsymbol{S}\mathbf{u}_t \end{bmatrix}$,
where $\boldsymbol{S}$ is the left-shift operator defined by 
$\boldsymbol{S}\mathbf{u}_{t} = \mathbf{u}_{t+1}$. 
The Koopman operator associated with the extended dynamics is defined on observables 
$\phi(\boldsymbol{\xi}_t)$ as 
$\mathcal{K}\phi(\boldsymbol{\xi}_t) \coloneqq \phi\big(f_{\boldsymbol{\xi}}(\boldsymbol{\xi}_t)\big)$.

The infinite-dimensional Koopman operator must be approximated in practice by a finite-dimensional representation. Several approaches have been proposed to obtain such approximations (see, e.g., \cite{williams2015data,williams2016extending,korda2018linear}), among which the data-driven Extended Dynamic Mode Decomposition (EDMD) is widely used. In EDMD, the set of observables is constructed through a ``lifted'' mapping 
$\phi(\boldsymbol{x}_{t},\boldsymbol{u}_{t}) = 
\left[\begin{array}{@{}c@{}}
        \psi(\boldsymbol{x}_{t}) \\
        \boldsymbol{u}_{t}
\end{array}\right]\!$, 
where $\psi:\mathcal{X}\to\mathbb{R}^{n_\psi}$ is a vector of observable (lifting) functions defined as 
$\psi(\boldsymbol{x}_t)\coloneqq\left[\psi_1(\boldsymbol{x}_t), \dots, \psi_{n_\psi}(\boldsymbol{x}_t)\right]^\top$, 
$n_\psi \gg n$ is the number of observables, 
and $\boldsymbol{u}_t \in \mathcal{U}$ denotes the input at time $t$. The functions in $\psi(\cdot)$ are typically chosen from a predefined set of basis functions, 
such as the radial basis functions used in \cite{korda2018linear}. A finite-dimensional approximation of the Koopman operator is then identified via an optimization procedure. Specifically, the identification of the approximate Koopman operator reduces to a least-squares problem, assuming that sampled data $\{(\boldsymbol{x}_{t}^j, \boldsymbol{u}_{t}^j),(\boldsymbol{x}_{t+1}^{j}, \boldsymbol{u}_{t+1}^{j})\}_{j=1}^{N_d}$ are collected 
according to the 
update mapping $\left[\begin{array}{@{}c@{}}
    \boldsymbol{x}_{t+1}^{j} \\
     \mathbf{u}_{t+1}^{j}
\end{array}\right]\! =\! \left[\begin{array}{@{}c@{}}
     f(\boldsymbol{x}_{t}^{j},\boldsymbol{u}_{t}^{j})  \\
     \boldsymbol{S} \mathbf{u}_{t}^{j}
\end{array}\right]$, where $j$ indexes the data samples and $t+1$ denotes the next time step, with each sample corresponding to a (possibly different) time index $t$. Although $\mathbf{u}_t$ is defined as a sequence, only its first element $\boldsymbol{u}_t$ is used in the data \cite{korda2018linear}. The approximate Koopman operator $\mathcal{A}$ is obtained by solving
\begin{equation}\label{eq: problem_EDMD_original} \min_{\mathcal{A}}\sum_{j=1}^{N_d}\|\phi(\boldsymbol{x}_{t+1}^{j},\boldsymbol{u}_{t+1}^{j})-\mathcal{A}\phi(\boldsymbol{x}_{t}^{j}, \boldsymbol{u}_{t}^{j})\|^2. 
\end{equation}
Since future control inputs do not need to be predicted, the last $q$ rows of $\mathcal{A}$ can be removed. Let $\bar{\mathcal{A}}$ denote the remaining part of $\mathcal{A}$ after removing the rows associated with the future control input. The matrix $\bar{\mathcal{A}}$ can then be decomposed as $\bar{\mathcal{A}}=[A\; B]$, where $A\in\rr^{n_\psi\times n_\psi}$ and $B\in\rr^{n_\psi\times q}$.
Problem \eqref{eq: problem_EDMD_original} can be reduced to solving $\min_{A,B} \sum_{j=1}^{N_d}\|\psi(\boldsymbol{x}_{t+1}^{j})-A\psi(\boldsymbol{x}_{t}^{j})-B\boldsymbol{u}_{t}^{j}\|^2$.
Additionally, the output matrix $C\in\rr^{n\times n_\psi}$ is obtained as the best projection of $\boldsymbol{x}$ onto the span of $\psi$ in a least-squares sense, i.e., as the solution to
\begin{equation}\label{eq: problem_EDMD_C}
     \min_C \sum_{j=1}^{N_d}\|\boldsymbol{x}_{t}^j-C\psi(\boldsymbol{x}_{t}^j)\|^2.
\end{equation}
At the end, a linear model for $\boldsymbol{x}$ can be formulated using
\begin{equation}\label{eq: out_Koopman_linear}
\psi_{t+1} = A \psi_t + B \boldsymbol{u}_t,~ \boldsymbol{x}_{t} \approx C \psi_{t},
\end{equation}
where $\psi_t\in\mathbb{R}^{n_\psi}$ denotes the lifted state propagated by the linear model, with $\psi_0=\psi(\boldsymbol{x}_0)$.
\subsection{Discrete-Time Control Barrier Functions (DCBFs)}
We consider a discrete-time control system \eqref{eq: sys} where function $f: \mathbb{R}^n \times \mathbb{R}^q \to \mathbb{R}^n$ is assumed to be locally Lipschitz continuous. Safety is defined in terms of forward invariance of a set $\mathcal{C}$. Specifically, system \eqref{eq: sys} is considered safe under an admissible control policy if, for any initial condition $\boldsymbol{x}_0 \in \mathcal{C}$, the state satisfies $\boldsymbol{x}_t \in \mathcal{C}$ for all $t \in \mathbb{N}$. We define $\mathcal{C}$ as the superlevel set of a continuous function $h: \mathbb{R}^n \to \mathbb{R}$:
\begin{equation}
\label{eq:safe-set}
\mathcal{C} \coloneqq \{ \boldsymbol{x} \in \mathbb{R}^n : h(\boldsymbol{x}) \geq 0 \}.
\end{equation}
\begin{definition}[DCBF~\cite{xiong2022discrete}]
\label{def:DCBF}
 Let $\mathcal{C}$ be defined by \eqref{eq:safe-set}. A continuous function $h:\mathbb{R}^{n}\to\mathbb{R}$ is a Discrete-time Control Barrier Function (DCBF) for system \eqref{eq: sys} if there exist class $\mathcal{K}$ function $\alpha$ satisfying $\alpha(r) < r$ for all $r > 0$ such that
 \begin{equation}
\Delta h(\boldsymbol{x}_{t},\boldsymbol{u}_{t}) \ge -\alpha(h(\boldsymbol{x}_{t})), \quad \forall \boldsymbol{x}_{t} \in \mathcal{C},\ \exists \boldsymbol{u}_t \in \mathcal{U}.
\label{eq:dh_condition}
\end{equation}
where $\Delta h(\boldsymbol{x}_{t},\boldsymbol{u}_{t}) = h(f(\boldsymbol{x}_t,\boldsymbol{u}_t)) - h(\boldsymbol{x}_{t})$. 
\end{definition}
 \begin{theorem}[Safety Guarantee~\cite{xiong2022discrete}]
\label{thm:forward-invariance}
Given a DCBF $h(\boldsymbol{x})$ from Def. \ref{def:DCBF} with the corresponding set $\mathcal{C}$ defined by \eqref{eq:safe-set}, if $\boldsymbol{x}_{0} \in \mathcal {C}$, then any Lipschitz controller $\boldsymbol{u}_{t}$ that satisfies the constraint in \eqref{eq:dh_condition}, $\forall t\ge 0$ renders $\mathcal {C}$ forward invariant for system \eqref{eq: sys}, $i.e., \boldsymbol{x}_{t} \in \mathcal {C}, \forall t\ge 0.$
\end{theorem}
We can simply define a DCBF $h(\boldsymbol{x})$ in \eqref{eq:dh_condition} satisfying
\begin{equation}
\label{eq:simple-DCBF}
\Delta h(\boldsymbol{x}_{t},\boldsymbol{u}_{t})
\ge -\gamma h(\boldsymbol{x}_{t}),
\end{equation}
where $\alpha(\cdot)$ is linear and $0<\gamma\le 1$. Rewriting \eqref{eq:simple-DCBF}, we have $h(f(\boldsymbol{x}_t,\boldsymbol{u}_t))\ge (1-\gamma)h(\boldsymbol{x}_{t})$. This shows 
the lower bound of $h$ decreases exponentially with the rate $1-\gamma$, and the DCBF turns into the discrete-time exponential
CBF \cite{agrawal2017discrete}.

\section{Problem Formulation}
\label{sec: problem and approach}
Our objective is to design a closed-loop control strategy for system \eqref{eq: sys} over the time interval $[0,T]$ that minimizes the deviation from a reference state, while satisfying safety requirements and state and input constraints.

\textbf{Safety Requirement:} System \eqref{eq: sys} should always satisfy a safety requirement of the form: 
\begin{equation}
\label{eq:Safety constraint}
h(\boldsymbol{x}_{t})\ge 0, ~\boldsymbol{x}_{t} \in \mathbb{R}^{n}, ~0\le t \le T,
\end{equation}
where $h:\mathbb{R}^{n}\to\mathbb{R}$. 

\textbf{Control and State Limitations:} The controller $\boldsymbol{u}_{t}$ and state $\boldsymbol{x}_{t}$ must satisfy the state and input constraints 
$\boldsymbol{x}_t\in\mathcal{X}\subset\mathbb{R}^n$, 
$\boldsymbol{u}_t\in\mathcal{U}\subset\mathbb{R}^q$ for $0\le t \le T.$

\textbf{Objective:} We consider the following cost:  
\begin{equation}
\label{eq:cost-function-1}
\begin{split}
J=
\left\|
\boldsymbol{x}_{t+N}-\boldsymbol{x}_{r}
\right\|_{P}^{2}
+
\sum_{k=0}^{N-1}
\Bigl(
\left\|
\boldsymbol{x}_{t+k}-\boldsymbol{x}_{r}
\right\|_{Q}^{2}
+\\
\left\|
\boldsymbol{u}_{t+k}-\boldsymbol{u}_{r}
\right\|_{R}^{2}
\Bigr),
\end{split}
\end{equation}
over a receding horizon $N<T$, where $\boldsymbol{x}_{t+k}$ and 
$\boldsymbol{u}_{t+k}$ denote the predicted state and input, 
$\boldsymbol{x}_{r}$ and $\boldsymbol{u}_{r}$ are the reference state and input, 
and $P$, $Q$, and $R$ are weighting matrices.

A control policy is said to be \emph{feasible} if there exists a control policy satisfying all state, input, and safety constraints for all $0\le t \le T$. In this paper, we consider the following problem:
\begin{problem}
\label{prob:Path-prob}
Find a feasible control policy for system \eqref{eq: sys} such that the safety requirement, control and state limitations are satisfied, and the cost \eqref{eq:cost-function-1} is minimized. 
\end{problem}
To enforce safety, the authors of \cite{zeng2021enhancing} incorporated the DCBF constraint
\begin{equation}
\label{eq: mpc-cbf-cons2}
\begin{split}
h(f(\boldsymbol{x}_{t+k},\boldsymbol{u}_{t+k})) \ge \omega_{t+k}(1-\gamma)h(\boldsymbol{x}_{t+ k}),\quad 0<\gamma\le1
\end{split}
\end{equation}
into the MPC formulation, where $\omega_{t+k}\ge 0$ is a relaxation variable introduced to improve both feasibility and safety. In addition, the system dynamics in \eqref{eq: sys} and the state and input constraints are included as constraints in the MPC problem to enforce additional requirements. 
However, since both \eqref{eq: sys} and \eqref{eq: mpc-cbf-cons2} are generally nonlinear, the resulting optimization problem becomes nonconvex. This approach, commonly referred to as NMPC-DCBF, may lead to high computational cost, particularly when long prediction horizons are considered. To transform the optimization problem into a convex one, the authors of \cite{liu2023iterative} linearize both the system dynamics \eqref{eq: sys} and the DCBF constraint \eqref{eq: mpc-cbf-cons2} as linear constraints within an MPC framework. The resulting linear MPC problem is solved iteratively at each time step to reduce the linearization error. However, the convergence of this iterative procedure is not guaranteed. To address this issue, we propose a Koopman-based MPC-DCBF framework to construct a linear MPC formulation without requiring iterative linearization.

\textbf{Approach:}
The key idea of this paper is to transform the nonlinear safety-constrained MPC problem into a linear MPC formulation using a Koopman-based lifted representation.
First, we augment the original system dynamics with the barrier function to construct a DCBF-augmented dynamical system that explicitly incorporates the safety constraint into the system state.
Second, the augmented nonlinear system is lifted into a higher-dimensional space through the Koopman operator framework, where the lifted dynamics can be approximated by a linear predictor identified from data.
This lifted representation enables both the system dynamics and the barrier function to admit linear expressions in the lifted state. To improve feasibility while preserving safety guarantees, a relaxation mechanism is introduced for the linearized DCBF constraints.
Finally, the resulting lifted linear model is integrated into an MPC formulation together with linearized DCBF constraints, which leads to a linear MPC problem that can be formulated as a QP solved efficiently at each control step.

\section{Koopman-based MPC-DCBF framework}
\label{sec: LMPC-DCBF}
In this section, we present the proposed Koopman-based MPC-DCBF framework, which enables the construction of a convex MPC problem for safety-critical systems. The framework consists of three main steps: constructing a DCBF-augmented dynamical system, lifting it to a higher-dimensional linear system via Koopman identification, and integrating the resulting lifted model into an MPC formulation.

\subsection{DCBF-augmented Dynamical System}
If only lifting system \eqref{eq: sys} to a higher and linear dimension, Problem \ref{prob:Path-prob} usually can not be transformed into a QP because of the nonlinear and possibly nonconvex safety constraint \eqref{eq:Safety constraint}. Even if $h$ is a $\ell_2$-norm DCBF, applying the DCBF constraint \eqref{eq:Safety constraint} to MPC results in a quadratically constrained QP (QCQP), which is computationally heavier than QP  in practice. 
Moreover, solving a QP not only has the same  $O(n^3)$ time complexity as solving a linear system of equations in theory \cite{11431115} but also offers data-independent execution time certificates \cite{10683964,11240592}, which is critical for real-time control applications.

To handle the general nonconvex DCBF constraint while maintaining computational efficiency, we incorporate the DCBF $h(\boldsymbol{x})$ as an additional state variable. Consider the discrete-time system \eqref{eq: sys}, where $h(\boldsymbol{x})$ denotes a DCBF encoding the safety constraint. The augmented state is defined as
\begin{equation}
\label{eq: aug state}
    \bar{\boldsymbol{x}}_t\coloneqq\left[\begin{array}{c}
         \boldsymbol{x}_t  \\
         h(\boldsymbol{x}_t) 
    \end{array} \right] \in\rr^{n+1},
\end{equation}
which is governed by the following DCBF-augmented dynamic system
\begin{equation}\label{eq: CBF_aug_sys}
    \bar{\boldsymbol{x}}_{t+1}=\bar{f}(\bar{\boldsymbol{x}}_t,\boldsymbol{u}_t)\coloneqq
    \left[\begin{array}{c}
         f(\boldsymbol{x}_t,\boldsymbol{u}_t)  \\
         h(f(\boldsymbol{x}_t,\boldsymbol{u}_t))
    \end{array} \right].
\end{equation}

Our idea is to apply the Koopman operator to lift the DCBF-augmented dynamical system \eqref{eq: CBF_aug_sys} into a linear system in a higher-dimensional space. This lifted representation enables both the system dynamics and the DCBF constraint to be expressed linearly in the lifted space, allowing the resulting MPC formulation to be cast as a QP.

\subsection{Koopman Identification via Least Squares}
\label{subec: Koopman Ide}
To obtain a linear representation of the nonlinear DCBF-augmented system \eqref{eq: CBF_aug_sys}, we employ the Koopman operator framework (EDMD in Sec. \ref{subsec: EDMD}) to lift the nonlinear system into a higher-dimensional space of observables where the dynamics can be approximated linearly.

Let $\psi(\bar{\boldsymbol{x}})$ denote a vector of observable functions of the augmented state $\bar{\boldsymbol{x}}$. 
The lifted state is initialized as
\begin{equation}\label{eq: lift state}
\psi_t = \psi(\bar{\boldsymbol{x}}_t) \in \mathbb{R}^{n_\psi},
\end{equation}
at the current time $t$, where $n_\psi$ is the dimension of the lifted space. 
Subsequent lifted states are propagated by an identified linear system
\begin{equation}
\label{eq: linear sys}
\psi_{t+1} = A \psi_t + B \boldsymbol{u}_t,
\end{equation}
where $A \in \mathbb{R}^{n_\psi \times n_\psi}$ and $B \in \mathbb{R}^{n_\psi \times q}$ are constant matrices to be identified. To compute $A$ and $B$, we collect data samples $\{(\bar{\boldsymbol{x}}_{t}^{j},\boldsymbol{u}_{t}^{j},\bar{\boldsymbol{x}}_{t+1}^{j})\}_{j=1}^{N_d}$
from the DCBF-augmented system \eqref{eq: CBF_aug_sys}, where $\bar{\boldsymbol{x}}_{t}^{j}$, $\boldsymbol{u}_{t}^{j}$, and $\bar{\boldsymbol{x}}_{t+1}^{j}$ denote the current augmented state, input, and next-step state, respectively. The corresponding lifted data are $\psi_t^j=\psi(\bar{\boldsymbol{x}}_{t}^{j})$, $\psi_{t+1}^j=\psi(\bar{\boldsymbol{x}}_{t+1}^{j})$. Stacking the data over $N_d$ samples yields
\begin{equation}
\Psi_{t+1} = A \Psi_{t} + B U_{t},
\end{equation} 
where $\Psi_{t} = [\psi_{t} ^1,\dots,\psi_{t} ^{N_d}]$, 
$\Psi_{t+1}= [\psi_{t+1}^1,\dots,\psi_{t+1}^{N_d}]$, 
$U_{t}= [\boldsymbol{u}_{t}^1,\dots,\boldsymbol{u}_{t}^{N_d}]$. Each sample $j$ corresponds to a (possibly different) time index $t$.

The matrices $A$ and $B$ are obtained by solving the following regularized least-squares problem:
\begin{equation}
\min_{A,B} \; \| \Psi_{t+1} - A\Psi_{t} - BU_{t} \|_F^2 + \lambda \|[A\;B]\|_F^2 ,
\end{equation}
where $\lambda>0$ is a regularization parameter and $\|\cdot\|_F$ denotes the Frobenius norm. This leads to the closed-form solution
\begin{equation}\label{eq: Mat A-B}
[A\; B] = \Psi_{t+1} V_{t}^\top (V_{t}V_{t}^\top + \lambda I)^{-1},
\end{equation}
where $V_{t} = [\Psi_{t}^\top \; U_{t}^\top]^\top$. The regularization term improves numerical stability and alleviates potential ill-conditioning of the data matrix.
Based on Eqn. \eqref{eq: problem_EDMD_C}, the optimal matrix $C$ admits the closed-form solution
\begin{equation}\label{eq: Mat C}
C = \bar{X}_{t} \Psi_{t}^\top (\Psi_{t} \Psi_{t}^\top)^{-1},
\end{equation}
where $\bar{X}_{t} = [\bar{\boldsymbol{x}}_{t}^1,\dots,\bar{\boldsymbol{x}}_{t}^{N_d}]$.
The resulting lifted linear model will be used in the next subsection to construct a linear MPC formulation.

\subsection{Koopman-based Linear MPC Formulation}
From Sec. \ref{subec: Koopman Ide}, the lifted linear dynamics obtained via the Koopman framework are given by Eq. \eqref{eq: linear sys}. Since the augmented state is defined by \eqref{eq: aug state}, the barrier value $h(\boldsymbol{x}_t)$ 
corresponds to the last component of $\bar{\boldsymbol{x}}_t$. 
Using the identified linear output map \eqref{eq: out_Koopman_linear}, \eqref{eq: lift state} and \eqref{eq: Mat C}, 
the barrier function can be expressed in the lifted space as
\begin{equation}
\label{eq: lifted bf}
h(\boldsymbol{x}_t)=e_h^\top\bar{\boldsymbol{x}}_t \approx e_h^\top C\psi_t \coloneqq \tilde{h}(\psi_t),
\end{equation}
where $e_h$ is the canonical basis vector selecting the last entry of $\bar{\boldsymbol{x}}_t$. 
Therefore, the barrier function is linear in the lifted state. The linear DCBF constraint defined by $\tilde{h}(\psi_t)$ can be formulated as a linear constraint to enforce the safety requirement. Together with the lifted linear dynamics \eqref{eq: linear sys} and the state and control bounds 
$\boldsymbol{x}_t\in\mathcal{X}\subset\mathbb{R}^n$ and 
$\boldsymbol{u}_t\in\mathcal{U}\subset\mathbb{R}^q$, 
the resulting optimization problem yields a linear MPC-DCBF formulation. 
Note that the state and control bounds may conflict with the system dynamics and DCBF constraints, which can lead to infeasibility. The authors of \cite{zeng2021enhancing} employed slack variables for DCBFs that enhance both feasibility and safety, as shown in \eqref{eq: mpc-cbf-cons2}. While this method enables the flexible integration of safety requirements as soft constraints within the control framework, the resulting constraints are nonlinear in the decision variables (see Remark \ref{rem:relaxation}). As a result, the optimization problem is generally nonconvex and inherently computationally intensive. To address this issue, we propose a new relaxation method that can generate linear DCBF constraints.

Consider a relaxed DCBF at current step $t$ as
\begin{equation}
\tilde{h}(\psi_{t+1})\geq \omega_{t}(1-\gamma) \tilde{h}(\psi_{t}),
\end{equation}
where $\omega_{t}\ge 0$ is a slack variable. Since $\tilde{h}(\psi_t)$ depends only on the current lifted state $\psi_t$, 
which is known at time $t$, it is treated as a known constant. To maintain linearity with respect to the decision variables, we propose
the following relaxation for the next prediction step:
\begin{equation}
\tilde{h}(\psi_{t+2})\geq \omega_{t+1}(1-\gamma)^{2} \tilde{h}(\psi_{t}), \ \omega_{t+1}\ge 0,
\end{equation}
which is a linear constraint in terms of $\psi_{t+2}$ and $\omega_{t+1}$. 
The same relaxation is imposed at each step over the prediction horizon. 
Having established the linearized DCBF, we proceed to link its construction to open-loop safety via the following theorem.
\begin{theorem}
\label{thm:safety-feasibility}
Given a linearized DCBF $\tilde{h}(\psi)$ for system \eqref{eq: linear sys}, if $\tilde{h}(\psi_{t})\ge 0$, then any Lipschitz controllers $\boldsymbol{u}_{t},...,\boldsymbol{u}_{t+N-1}$ that satisfy the constraints $\tilde{h}(\psi_{t+k+1})\geq \omega_{t+k}(1-\gamma)^{k+1} \tilde{h}(\psi_{t}), \ 0<\gamma \le 1, \ \omega_{t+k}\ge 0,\ k\in \{0,...,N-1\}$ render $\tilde{h}(\psi_{t+k+1})\ge 0$.
\end{theorem}
\begin{proof}
Since $\tilde{h}(\psi_t)\ge 0$, $\omega_{t+k}\ge 0$, and
$0<\gamma\le 1$, we have
$\omega_{t+k}(1-\gamma)^{k+1}\tilde{h}(\psi_t)\ge 0,
\ k\in\{0,\ldots,N-1\}$.
Therefore, from the imposed constraint,
$\tilde{h}(\psi_{t+k+1})
\ge
\omega_{t+k}(1-\gamma)^{k+1}\tilde{h}(\psi_t)
\ge 0$,
for all $k\in\{0,\ldots,N-1\}$.
\end{proof}
From the above proof, we observe that although a slack variable $\omega_{t+k}$ is introduced into the constraint, the open-loop safety (i.e., $\tilde{h}(\psi_{t+k+1})\ge 0,\ k\in \{0,...,N-1\}$) is still guaranteed. Moreover, when no control input satisfies $\tilde{h}(\psi_{t+k+1})\geq \omega_{t+k}(1-\gamma)^{k+1}\tilde{h}(\psi_t)$ for $\omega_{t+k}\geq 1$, the optimizer may select $\omega_{t+k}\in[0,1)$, thereby relaxing the right-hand side of the constraint to improve feasibility. This relaxation method enables us to incorporate DCBF
constraints in a linear MPC (LMPC) formulation:

\noindent\rule{\columnwidth}{0.4pt}
\textbf{LMPC-DCBF:}
{\small
\begin{subequations}\label{eq: LMPC_DCBF}
\begin{align}
\min_{U_{t},\Psi_{t},\Omega_{t}} \ &p(\psi_{t+N})+\sum_{k=0}^{N-1} q(\psi_{t+k},\boldsymbol{u}_{t+k},\omega_{t+k}) \label{eq: cost2}\\
\text{s.t.} ~~
 \psi&_{t+k+1} =  A \psi_{t+k} + B\boldsymbol{u}_{t+k}, \ k=0,\dots,N-1, \label{eq: dynamic_constraint2}  \\
 E_x C&\psi_{t+k} \in \mathcal{X}, \ k=1,\dots,N, \label{eq: state_constraint2} \\
 \boldsymbol{u}_{t+k} &\in \mathcal{U}, \ k=0,\dots,N-1 ,\label{eq: input_constraint2} \\
 \tilde{h}(\psi&_{t+k+1})\geq \omega_{t+k}(1-\gamma)^{k+1} \tilde{h}(\psi_{t}), \ 0<\gamma \le 1, \label{eq: cbf_constraint2}\\ \omega_{t+k}&\ge 0, \  k=0,\dots,N-1, 
\end{align}
\end{subequations}
}
\noindent\rule{\columnwidth}{0.4pt}
where $U_{t} = [\boldsymbol{u}_{t}, \dots, \boldsymbol{u}_{t+N-1}]$, $\Psi_{t} = [\psi_{t}, \dots, \psi_{t+N}]$ and $\Omega_{t} = [\omega_{t}, \dots, \omega_{t+N-1}]$. $E_x=[\mathcal{I}_n\;\mathbf{0}]$ is a selection matrix that extracts the original state 
$\boldsymbol{x}_t$ from the augmented state $\bar{\boldsymbol{x}}_t$. Thus, $\psi_{t+k}$ in the cost function \eqref{eq: cost2} can be converted to 
$\boldsymbol{x}_{t+k}\approx E_x C\psi_{t+k}$, enabling the cost to penalize the 
deviation from the desired state $\boldsymbol{x}_r$ in \eqref{eq:cost-function-1}.
In practice, \( \omega_{t+k} \) is treated as an optimization variable that is penalized toward 1---either through a large weight in the cost function or via a lexicographic priority. As a result, when the hard DCBF constraint is feasible, the optimizer yields \( \omega^{\star}_{t+k} \approx 1 \) (i.e., no relaxation). Otherwise, it computes a minimally relaxed solution that still ensures \( \tilde{h}(\psi_{t+k+1}) \ge 0 \).

Based on \cite[Def. 1]{liu2023iterative}, if $\tilde{h}(\psi_{t})$ has relative degree $m$ and the constraints in an optimal control problem consist solely of DCBF constraints \eqref{eq: cbf_constraint2}, then the DCBFs must be formulated up to the \( m^\text{th} \) order to ensure that the control input \( \boldsymbol{u}_{t} \) appears explicitly, i.e., the control input \( \boldsymbol{u}_{t} \) will appears in $\tilde{h}(\psi_{t+m},\boldsymbol{u}_{t})$. However, if the framework also includes the system dynamics \eqref{eq: dynamic_constraint2} as constraints---where \( \boldsymbol{u}_{t} \) is already explicitly represented---then the order of the DCBFs can be chosen more flexibly to reduce computational complexity. In particular, if \( m \le N \), the DCBF order can be set to 1, since the control input will affect the DCBF constraints within the prediction horizon. In our LMPC-DCBF framework \eqref{eq: LMPC_DCBF}, since the barrier function is represented in the lifted space by \eqref{eq: lifted bf},
its relative degree with respect to the lifted linear system \eqref{eq: dynamic_constraint2} is determined by the smallest integer $m$ such that
\begin{equation}
e_h^\top C A^{m-1} B \neq \mathbf{0},
\end{equation}
while $e_h^\top C A^i B = \mathbf{0},~ i=0,\dots,m-2$. Thus, when identifying $A$, $B$, and $C$ in \eqref{eq: Mat A-B}–\eqref{eq: Mat C}, the relative degree $m$ of the barrier function with respect to the lifted system cannot be predetermined while preserving the accuracy of the lifted dynamics. Therefore, the prediction horizon $N$ should be chosen sufficiently large to ensure $m \le N$.

\begin{remark}
\label{rem: conservative DCBF}
The Koopman-based lifted model in \eqref{eq: dynamic_constraint2} is obtained from a finite-dimensional approximation of the Koopman operator using data-driven identification. As a result, the lifted linear dynamics only approximate the original nonlinear system and modeling errors may exist. Therefore, the safety guarantees derived from the linear MPC-DCBF formulation cannot strictly ensure safety of the original nonlinear system in a rigorous sense. 
In practice, this issue can be mitigated through several conservative design choices. For example, a safety margin can be incorporated into the barrier function, i.e., enforcing $\tilde{h}(\psi_{t}) \ge \epsilon$ with $\epsilon>0$, to compensate for approximation errors. In addition, the class $\mathcal{K}$ function hyperparameter $\gamma$ in the DCBF constraint can be adjusted. Choosing a smaller $\gamma$ leads to more conservative constraints and strengthens the safety enforcement. In this paper, we adopt the latter strategy and regulate the safety level by tuning $\gamma$, as illustrated in the numerical examples in Sec.~\ref{sec: case studies}.
\end{remark}
\begin{remark}
\label{rem:relaxation}
In \cite{zeng2021enhancing}, the relaxation technique places the slack variable 
$\omega_{t+k}$ in front of $(1-\gamma)h(\boldsymbol{x}_{t+k})$, yielding
\begin{equation}
\label{eq: mpc-cbf-cons22}
\begin{split}
h(f(\boldsymbol{x}_{t+k},\boldsymbol{u}_{t+k})) \ge \omega_{t+k}(1-\gamma)h(\boldsymbol{x}_{t+ k}), \omega_{t+k}\ge0.
\end{split}
\end{equation}
However, $\omega_{t+k}$ is a decision variable, and $\boldsymbol{x}_{t+k}$ is also a decision variable for $k \ge 1$ in the MPC problem. As a result, the product $\omega_{t+k}h(\boldsymbol{x}_{t+k})$ introduces a nonlinear constraint, which leads to a nonconvex optimization problem.

To address this issue, in Theorem~\ref{thm:safety-feasibility} we establish a recursive relationship that connects the future barrier values $\tilde{h}(\psi_{t+k+1})$ to the known constant $\tilde{h}(\psi_t)$. 
By multiplying the slack variable $\omega_{t+k}$ with this known quantity, we obtain the constraint in \eqref{eq: cbf_constraint2}, which is linear with respect to the decision variables.
\end{remark}
\subsection{Complexity Analysis}
When $p(\cdot)$ and $q(\cdot)$ in \eqref{eq: cost2} are chosen as quadratic cost functions and all constraints in \eqref{eq: dynamic_constraint2}--\eqref{eq: cbf_constraint2} are affine in the decision variables, the resulting optimization problem is a convex QP at each time step.
The decision variables consist of the predicted control inputs 
$U_t = [\boldsymbol{u}_t, \dots, \boldsymbol{u}_{t+N-1}]$, lifted states 
$\Psi_t = [\psi_t, \dots, \psi_{t+N}]$, and relaxation variables 
$\Omega_t = [\omega_t, \dots, \omega_{t+N-1}]$. 
Since the lifted state dimension $n_\psi$, the input dimension $q$, and the 
prediction horizon $N$ have been defined previously, the number of decision 
variables in the LMPC-DCBF problem \eqref{eq: LMPC_DCBF} scales on the order of $ O\!\left(N(n_\psi + q)\right)$.
Here we omit the contribution of the relaxation variables $\omega_{t+k}$ 
since each of them is a scalar and thus does not affect the asymptotic complexity order.
Using standard interior-point methods, the worst-case computational complexity scales as $
O\!\left((N(n_\psi + q))^3\right)$.
In practice, structure-exploiting QP solvers (e.g., OSQP) typically achieve significantly faster performance.

Compared with \cite{zeng2021enhancing}, \cite{liu2023iterative}, the proposed Koopman-based LMPC-DCBF framework maintains a convex QP structure and avoids iterative linearization steps within each control update, which substantially reduces the computational burden and facilitates real-time implementation.

\section{Numerical Examples}
\label{sec: case studies}
In this section, we present numerical results to validate our proposed approach using a unicycle model. We compare the performance of the proposed method with the baseline NMPC-DCBF and iMPC-DCBF approaches.  
For the proposed method, the QPs were solved using \texttt{quadprog}. 
For iMPC-DCBF, OSQP~\cite{stellato2020osqp} was used to solve the convex optimization problem at each iteration. 
NMPC-DCBF was solved using IPOPT~\cite{biegler2009large} with the modeling language YALMIP~\cite{lofberg2004yalmip}. 
All simulations were conducted on a Windows desktop with an Intel® Core™ i7-11750F CPU @ 2.50 GHz.

\subsection{System Dynamics and DCBF}
Consider a continuous-time unicycle model in the form
\begin{equation}
\label{eq:unicycle-model}
\begin{bmatrix}
\dot{x}(t) \\
\dot{y}(t) \\
\dot{\theta}(t)\\
\dot{v}(t)
\end{bmatrix}  
=\begin{bmatrix}
 v(t)\cos{(\theta(t))}  \\
 v(t)\sin{(\theta(t))} \\
 0 \\
 0
\end{bmatrix} 
+ \begin{bmatrix}
  0 & 0\\
  0 & 0\\
  1 & 0\\
  0 & 1 
\end{bmatrix}\begin{bmatrix}
   u_{1}(t)   \\
  u_{2}(t) 
\end{bmatrix},
\end{equation}
where $\boldsymbol{x}(t)=[x(t),y(t),\theta(t),v(t)]^{\top}$ captures the 2-D location, heading angle, and linear speed; $\boldsymbol{u}(t)=[u_1(t),u_{2}(t)]^{\top}$ represents angular velocity ($u_{1}$) and linear acceleration ($u_{2}$), respectively. The discrete-time model used in this paper is obtained by discretizing \eqref{eq:unicycle-model} using a standard fourth-order Runge--Kutta scheme with
sampling time $\Delta t$, yielding the discrete-time model in \eqref{eq: sys}. 

As a candidate DCBF function $h(\boldsymbol{x}_{t})$, we choose a quadratic distance function for circular obstacle avoidance $h(\boldsymbol{x}_{t})= (x_{t}-x_{0})^{2}+(y_{t}-y_{0})^{2}-r^{2}$, where $(x_{0},y_{0})$ and $r$ denote the obstacle center location and radius, respectively. In this work, we set $x_{0}=0$, $y_{0}=0$, and $r=1$.

\subsection{Koopman Identification}
\label{subsec: Koopman model}
To identify the lifted linear predictor, we generate state-transition
data $(\boldsymbol{x}_t,\boldsymbol{u}_t,\boldsymbol{x}_{t+1})$ using the
discrete-time model. Random initial states and control inputs are
sampled within the admissible ranges to construct the dataset used
for regression.
To handle the periodicity of the heading angle, the lifted state
does not directly include $\theta_t$. Instead, we use the
trigonometric coordinates $(\sin\theta_t,\cos\theta_t)$.
The lifted state includes the physical states
$[x_t,\,y_t,\,v_t,\,\sin\theta_t,\,\cos\theta_t]^\top$,
the dynamics-consistent terms
$[v_t\cos\theta_t,\,v_t\sin\theta_t]^\top$,
and the candidate DCBF function $h(\boldsymbol{x}_t)$.
To further improve approximation accuracy, 200 Gaussian radial basis
functions (RBFs) \cite{williams2015data} are included, resulting in the
lifted state $\psi_t$ with dimension $n_\psi=208$.

Using the collected dataset and the EDMD formulation introduced in
Sec.~\ref{subec: Koopman Ide}, we identify the lifted linear dynamics
in \eqref{eq: linear sys}. 
The dataset contains $N_d=2\times10^5$ transition samples generated
from randomly sampled initial states and inputs within the ranges
$x,y\in[-3,3]$, $\theta\in[-\pi,\pi]$, $v\in[-3,3]$, and
$u_1,u_2\in[-3,3]$.
The physical states are reconstructed from the leading lifted states,
while the heading angle is recovered from $(\sin\theta_t,\cos\theta_t)$
using $\mathrm{atan2}$.

Fig.~1 illustrates the prediction performance of the identified Koopman model. 
The initial state is set to $\boldsymbol{x}_0=[-3,-3,0,0.2]^{\top}$. 
To distinguish it from the true state $\boldsymbol{x}_t$ generated by the nonlinear dynamics, 
we denote by $\hat{\boldsymbol{x}}_t=E_x C\psi_{t}$ the state reconstructed from the Koopman predictor. 
Figs.~\ref{subfig:1}–\ref{subfig:2} show the trajectories obtained from the true system and the Koopman model 
under sampling times $\Delta t=0.1$s and $\Delta t=0.01$s, respectively. 
Figs.~\ref{subfig:3}–\ref{subfig:4} compare the obstacle function values, where 
$h(\boldsymbol{x}_t)$ denotes the true barrier function evaluated at the true state (nonlinear), 
$h(\hat{\boldsymbol{x}}_t)$ is computed using the reconstructed Koopman state (nonlinear), and 
$\tilde{h}(\psi_t)$ is obtained directly from the lifted state (linear). 
Fig.~\ref{subfig:5} shows the prediction error $e_t=\|\boldsymbol{x}_t-\hat{\boldsymbol{x}}_t\|_2$, which measures the discrepancy between the true state $\boldsymbol{x}_t$ 
and the state reconstructed from the Koopman lifted representation 
via $\hat{\boldsymbol{x}}_t=E_x C \psi_t$.
It can be observed that the Koopman predictor closely matches the true system behavior, 
and the approximation accuracy improves as the sampling time $\Delta t$ decreases.
\begin{figure*}[!t]
    \vspace*{0.2cm}
    \centering
    \begin{subfigure}[t]{0.19\linewidth}
        \centering
    \includegraphics[width=\linewidth]{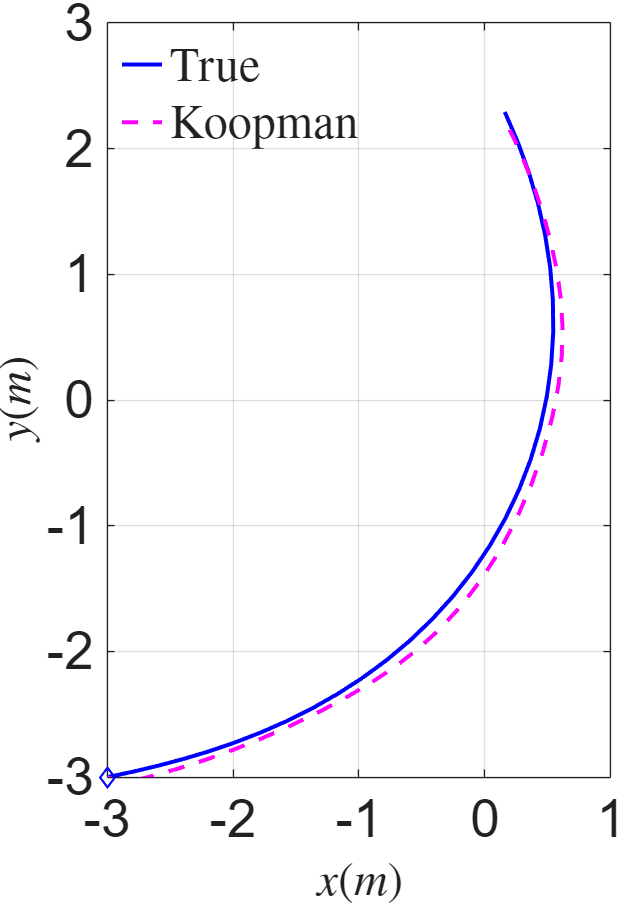}
        \caption{Trajectory $\Delta t=0.1s$}
        \label{subfig:1}
    \end{subfigure}    
    \begin{subfigure}[t]{0.19\linewidth}
        \centering
    \includegraphics[width=\linewidth]{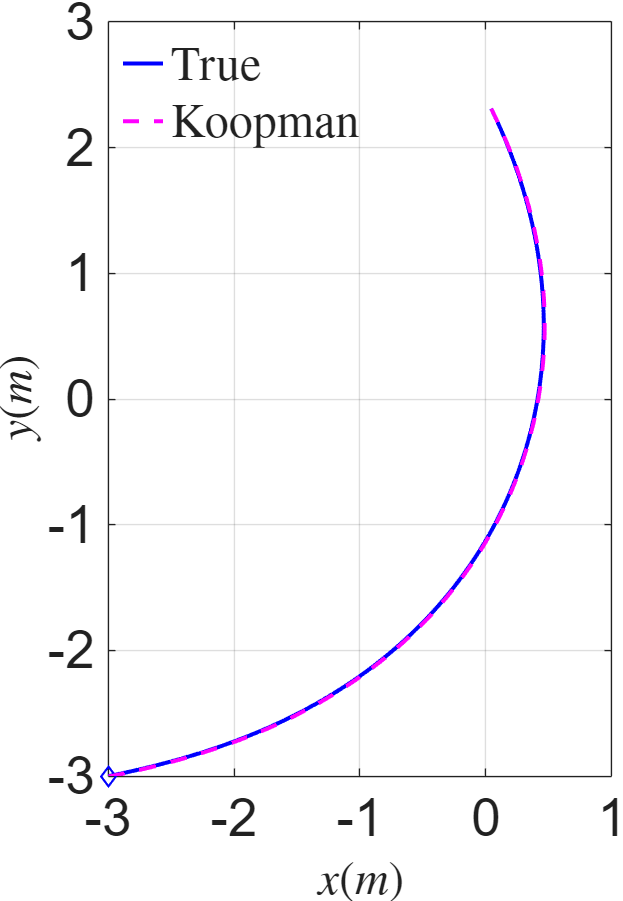}
        \caption{Trajectory $\Delta t=0.01s$}
        \label{subfig:2}
    \end{subfigure}     
    \begin{subfigure}[t]{0.173\linewidth}
        \centering
    \includegraphics[width=\linewidth]{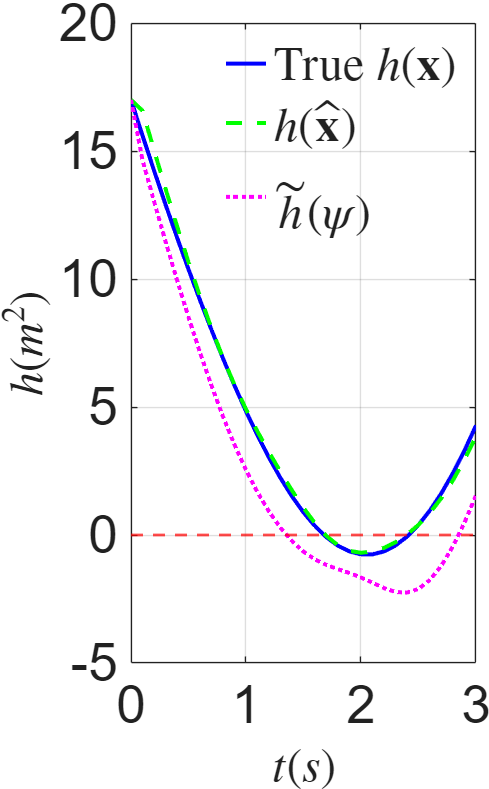}
        \caption{DCBFs $\Delta t=0.1s$}
        \label{subfig:3}
    \end{subfigure}  
        \begin{subfigure}[t]{0.173\linewidth}
        \centering
      \includegraphics[width=\linewidth]{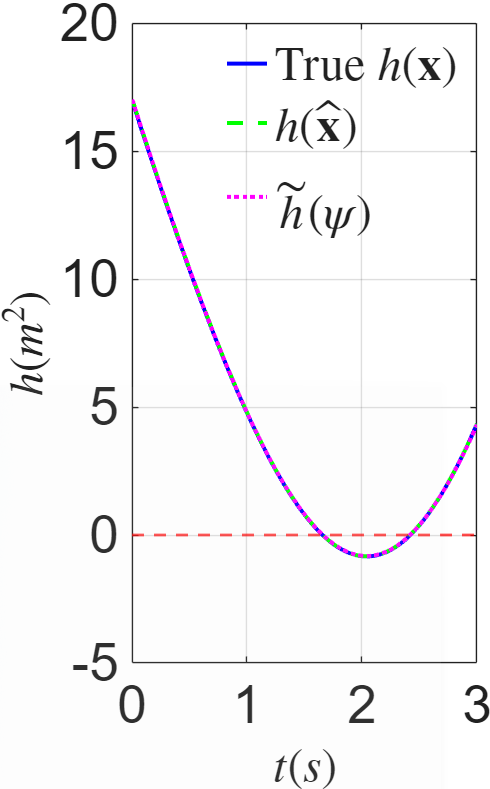}
        \caption{DCBFs $\Delta t=0.01s$}
        \label{subfig:4}
    \end{subfigure}
        \begin{subfigure}[t]{0.2\linewidth}
        \centering
      \includegraphics[width=\linewidth]{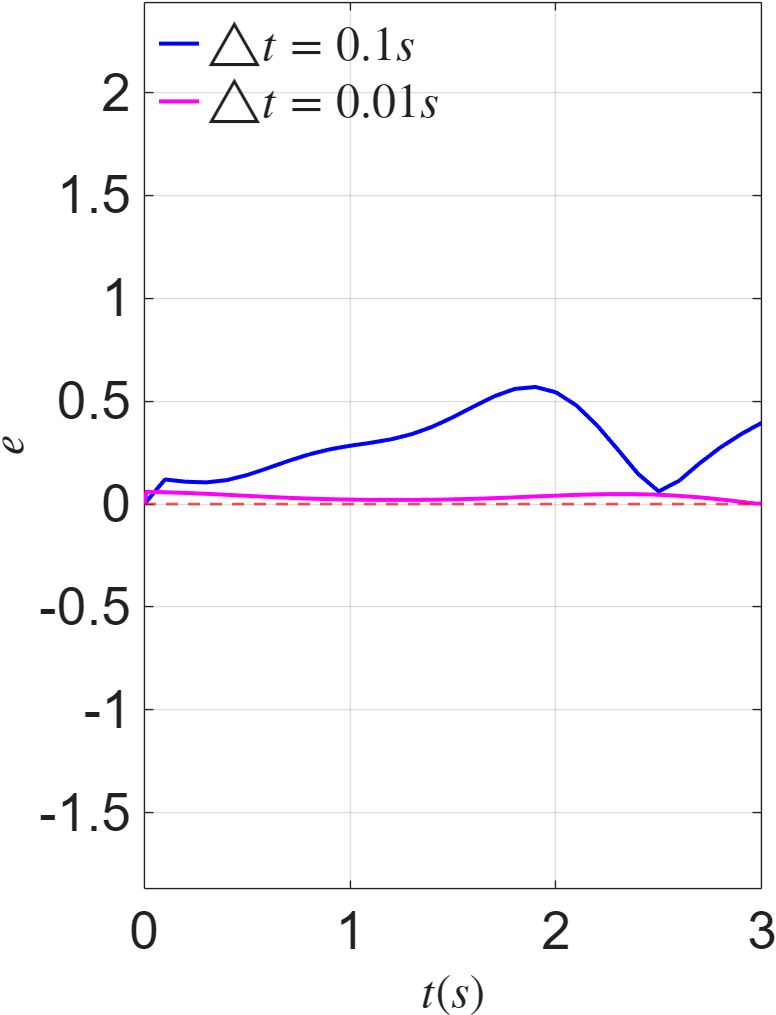}
        \caption{Error}
        \label{subfig:5}
    \end{subfigure}
    \caption{Comparison between the true nonlinear dynamics and the identified Koopman predictor. 
(a), (b): $x$--$y$ trajectories generated by the true system and the Koopman model. 
(c), (d): evolution of the DCBFs. 
(e): prediction error $e_t=\|\boldsymbol{x}_t-\hat{\boldsymbol{x}}_t\|_2$.}
    \label{fig: Prediction}
\end{figure*}

\subsection{MPC Design}
The cost function of the MPC problem consists of the stage cost
$q(\psi_{t+k},\boldsymbol{u}_{t+k},\omega_{t+k})
=\|C_y\psi_{t+k}-\boldsymbol{y}_{r}\|_Q^2+
\|\boldsymbol{u}_{t+k}-\boldsymbol{u}_{r}\|_R^2+
\|\omega_{t+k}-\omega_{r}\|_S^2$
and the terminal cost
$p(\psi_{t+N})=
\|C_y\psi_{t+N}-\boldsymbol{y}_{r}\|_P^2$,
where
$C_y\psi_t=
[x_t,y_t,\sin\theta_t,\cos\theta_t,v_t]^\top$ and
$\boldsymbol{y}_{r}=
[x_r,y_r,\sin\theta_r,\cos\theta_r,v_r]^\top$. The weighting matrices are
$Q=P=\mathrm{diag}(500,500,80,80,10)$,
$R=\mathrm{diag}(0.3,0.5)$, and $S=1000$.
The initial and reference states are
$\boldsymbol{x}_{0}=[-2.5,-2,0.2,2]^\top$
and
$\boldsymbol{x}_{r}=[2,2.5,0,0.3]^\top$, respectively.
The other references are
$\boldsymbol{u}_{r}=[0,0]^\top$ and $\omega_r=1$.
System \eqref{eq:unicycle-model} is discretized with
$\Delta t=0.01$ and is subject to the following state and input constraints:
\begin{equation}
\begin{split}
\label{eq:state-input-constraint}
\mathcal{X}&=\{\boldsymbol{x}_{t}\in \mathbb{R}^{4}: -3\cdot \mathcal{I}_{4\times1} \le \boldsymbol{x}_{t}\le 3\cdot \mathcal{I}_{4\times1}\},\\
\mathcal{U}&=\{\boldsymbol{u}_{t}\in \mathbb{R}^{2}: -3\cdot \mathcal{I}_{2\times1} \le \boldsymbol{u}_{t}\le 3\cdot \mathcal{I}_{2\times1}\}.
\end{split}
\end{equation}

We incorporate the Koopman model identified in Sec.~\ref{subsec: Koopman model} into the MPC framework to solve the optimization problem in \eqref{eq: LMPC_DCBF}. The resulting controller is referred to as K-LMPC-DCBF. If the DCBF constraint \eqref{eq: cbf_constraint2} is removed from the optimization, the controller reduces to K-LMPC. As benchmarks, we compare with iMPC-DCBF~\cite{liu2023iterative}, where the maximum iteration number is set to $j_{\max}=1000$, and NMPC-DCBF~\cite{zeng2021enhancing}. To ensure a fair comparison, the parameters of all methods are kept identical.

As shown in Fig.~\ref{fig: 6}, without the DCBF constraint, K-LMPC enters the obstacle region and therefore violates safety. In contrast, K-LMPC-DCBF is able to enforce safety during the motion. It can also be observed that as the prediction horizon $N$ increases, the trajectory of K-LMPC-DCBF approaches the obstacle boundary. This occurs because a longer prediction horizon allows the optimizer to choose a more path-optimal control strategy.

For the case $N=32$ and $\gamma=0.4$, K-LMPC-DCBF slightly penetrates the obstacle region. This is because the controller attempts to optimize the path while the Koopman-lifted dynamics and DCBF constraints are only approximations of the true system, which introduces modeling errors. However, when $\gamma$ is reduced to $0.2$, a more conservative DCBF condition is enforced and the trajectory becomes safe again, which is consistent with the analysis in Remark~\ref{rem: conservative DCBF}.

Both iMPC-DCBF and NMPC-DCBF can guarantee safety for the case $N=32$, $\gamma=0.4$. However, the computational advantage of K-LMPC-DCBF is significant. The average control computation time of K-LMPC-DCBF is about $3$--$5$ ms per step, while iMPC-DCBF requires approximately $100$--$330$ ms, which is about $30$--$60$ times slower. NMPC-DCBF is the slowest, taking about $10500$--$16700$ ms per step (around $3500$ times slower). These results demonstrate the substantial computational efficiency of the proposed K-LMPC-DCBF framework.
\begin{figure}
\vspace{3mm}
    \centering
    \includegraphics[scale=0.22]{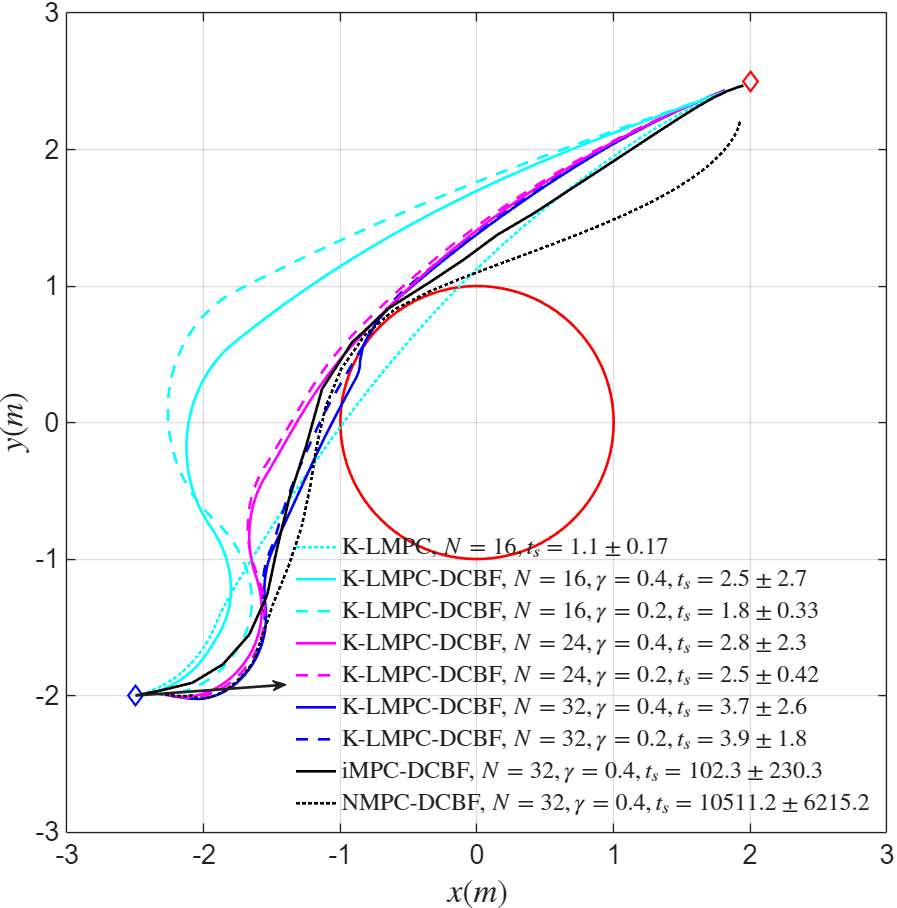}
    \caption{Trajectory comparison in the presence of a circular obstacle.
The robot moves from the start (blue diamond) to the goal (red diamond).
The black arrow indicates the initial velocity direction of the robot.
Trajectories generated by the proposed K-LMPC-DCBF with different
prediction horizons $N$ and hyperparameters $\gamma$ are compared with
iMPC-DCBF and NMPC-DCBF. The legend also reports the average computation
time per step $t_s$ (ms).}
    \label{fig: 6}
\end{figure}

To better visualize safety, Fig.~\ref{fig: 7} shows the DCBF value over time for all methods. 
For K-LMPC-DCBF, when $\gamma$ is fixed, increasing the prediction horizon $N$ decreases the minimum value of $h$, bringing the trajectory closer to the obstacle boundary. 
Conversely, for fixed $N$, a smaller $\gamma$ results in a larger minimum value of $h$, indicating a more conservative and safer behavior. 
This again demonstrates the safety guarantee provided by K-LMPC-DCBF.
\begin{figure}
\vspace{3mm}
    \centering
    \includegraphics[scale=0.225]{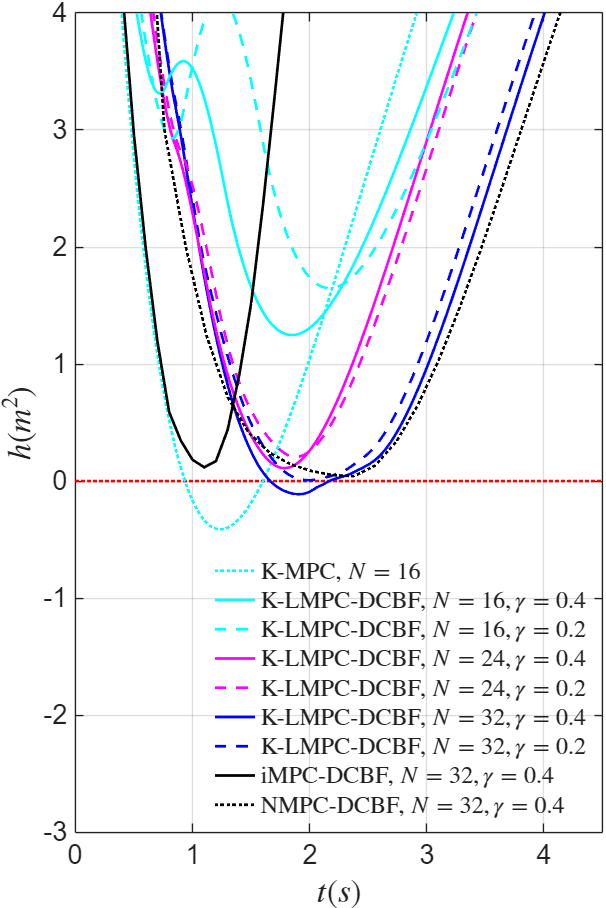}
\caption{Evolution of the DCBFs under different control strategies. 
The red dashed line indicates the safety boundary.}
    \label{fig: 7}
\end{figure}

Fig.~\ref{fig: input} shows the control inputs generated by K-LMPC-DCBF under different hyperparameter settings. 
The inputs remain within the prescribed bounds throughout the motion, confirming that the control constraints are respected. 
At the beginning of the maneuver, relatively large control actions are required to steer the robot away from the obstacle. 
Some nonsmooth variations can be observed during this stage when the safety constraints become active and the inputs approach their bounds, which is typical for constrained MPC-based safety controllers. 
As the robot moves away from the obstacle and approaches the goal, the control inputs gradually become smoother and converge toward zero. 
\begin{figure}[!t]
    \vspace*{0.2cm}
    \centering
        \begin{subfigure}[t]{0.48\linewidth}
        \centering
      \includegraphics[width=\linewidth]{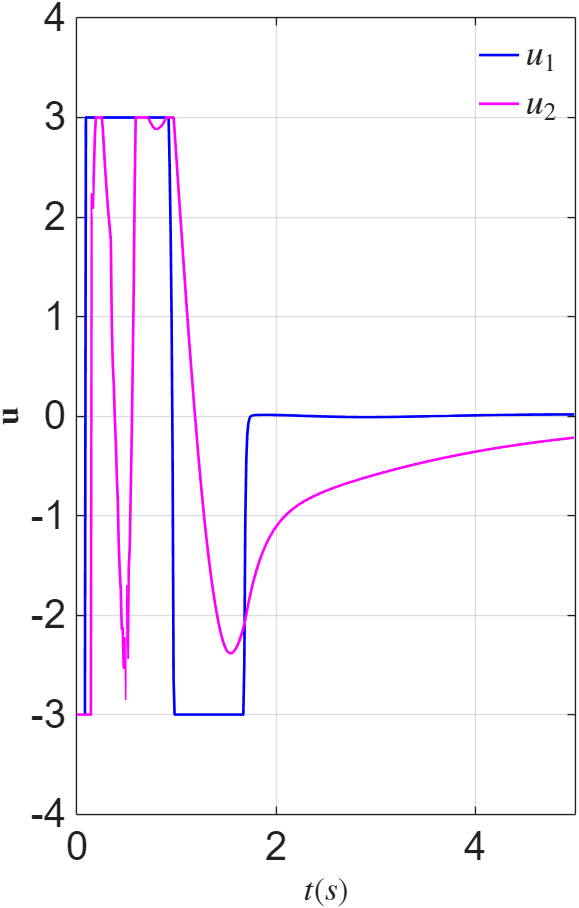}
        \caption{$N=16, \gamma=0.2$}
        \label{subfig:8}
    \end{subfigure}
        \begin{subfigure}[t]{0.48\linewidth}
        \centering
      \includegraphics[width=\linewidth]{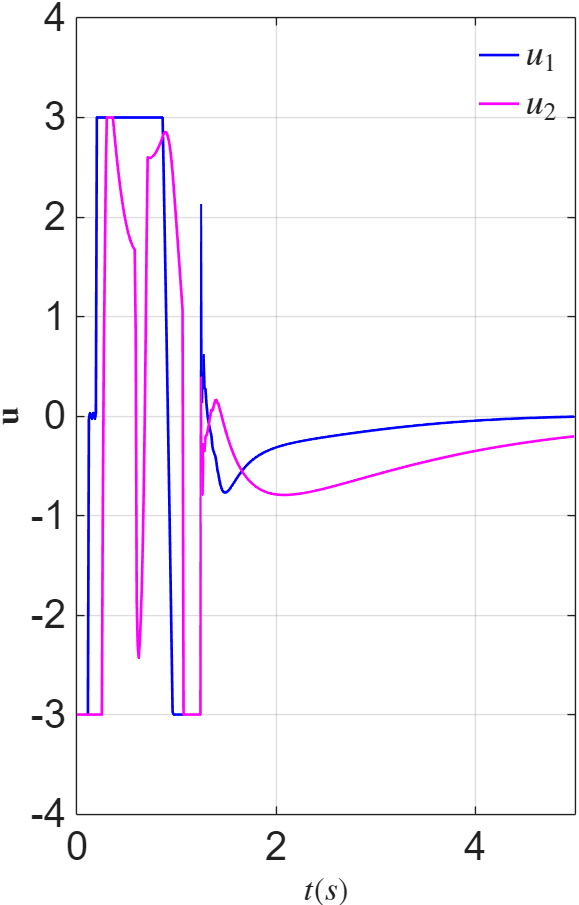}
        \caption{$N=24, \gamma=0.2$}
        \label{subfig:9}
    \end{subfigure}
    \caption{Control inputs $u_1$ and $u_2$ over time for K-LMPC-DCBF under different hyperparameter settings.}
    \label{fig: input}
\end{figure}

\section{Conclusion}
\label{sec:conclusion}
This paper proposed a Koopman-based MPC-DCBF framework for safety-critical control of nonlinear discrete-time systems. 
By incorporating the barrier function into the lifted state and identifying a linear predictor using Koopman operator theory, the nonlinear safety-constrained MPC problem can be transformed into a QP. This formulation enables efficient real-time computation, while preserving safety guarantees through DCBF constraints. 

Numerical simulations on a mobile robot navigation task demonstrated that the proposed method achieves safe trajectory generation and significantly reduces computational cost compared with NMPC-DCBF and iterative MPC-DCBF approaches.
Future work will focus on extending the framework to more complex environments, addressing model uncertainty, and applying the method to hardware experiments and multi-agent systems.

\bibliographystyle{IEEEtran}
\bibliography{ref} 
\end{document}